\documentclass[aps,superscriptaddress,twocolumn,twoside,floatfix,prl,a4paper]{revtex4-2} %nofootinbib

\usepackage{times}
\usepackage{epsfig}
\usepackage{amsfonts}
\usepackage{amsmath} 
\usepackage{amssymb,amsthm}
\usepackage{xcolor,colortbl}
\usepackage{multirow}
\usepackage{braket}
\usepackage{latexsym}
\usepackage{tabularx}
\usepackage{amsfonts}
\usepackage{mathrsfs}
\usepackage{natbib}
\usepackage{verbatim}
\usepackage{gensymb}
\usepackage{caption}
\usepackage{caption}
\usepackage{subcaption}
\usepackage{ragged2e}
\usepackage{ulem}
\DeclareCaptionJustification{justified}{\justifying}
\usepackage{blkarray}
 \usepackage{graphicx}
 \usepackage[export]{adjustbox}
 \newcommand{\ba}{\begin{eqnarray}}
\newcommand{\ea}{\end{eqnarray}}
\newcommand{\ketbra}[2]{|#1\rangle \langle #2|}
\newcommand{\simm}{\hspace{-.1cm}\sim\hspace{-.1cm}}
\newcommand{\nsimm}{\hspace{-.1cm}\nsim\hspace{-.1cm}}
\newcommand{\suc}{\hspace{-.1cm}\succ\hspace{-.1cm}}
\newcommand{\ansimm}{\hspace{-.1cm}\underset{asy}{\nsim}\hspace{-.1cm}}

\newtheorem{theorem}{Theorem}
\newtheorem{corollary}{Corollary}

\newtheorem{proposition}{Proposition}

\usepackage[colorlinks=true,linkcolor=blue,citecolor=magenta,urlcolor=blue]{hyperref}
\allowdisplaybreaks

\newtheorem{Lemma}{Lemma}

\begin{document}

%\title{Asymptotic Inequivalence of Nonlocal Quantum Correlations} 

\title{Quantum Nonlocality: Multi-copy Resource Inter-convertibility \& Their Asymptotic Inequivalence}

\author{Subhendu B. Ghosh}
\affiliation{Physics and Applied Mathematics Unit, 203 B.T. Road Indian Statistical Institute Kolkata, 700108}

\author{Snehasish Roy Chowdhury}
\affiliation{Physics and Applied Mathematics Unit, 203 B.T. Road Indian Statistical Institute Kolkata, 700108}

\author{Guruprasad Kar}
\affiliation{Physics and Applied Mathematics Unit, 203 B.T. Road Indian Statistical Institute Kolkata, 700108}

\author{Arup Roy}
\affiliation{Department of Physics, A B N Seal College Cooch Behar, West Bengal 736101, India}

\author{Tamal Guha}
\affiliation{QICI Quantum Information and Computation Initiative, Department of Computer Science, The University of Hong Kong, Pokfulam Road, Hong Kong}

\author{Manik Banik}
\affiliation{Department of Physics of Complex Systems, S.N. Bose National Center for Basic Sciences, Block JD, Sector III, Salt Lake, Kolkata 700106, India.}

\begin{abstract}
Quantum nonlocality, pioneered in Bell's seminal work and subsequently verified through a series of experiments, has drawn substantial attention due to its practical applications in various protocols. Evaluating and comparing the extent of nonlocality within distinct quantum correlations holds significant practical relevance. Within the resource theoretic framework this can be achieved by assessing the inter-conversion rate among different nonlocal correlations under free local operations and shared randomness. In this study we, however, present instances of quantum nonlocal correlations that are incomparable in the strongest sense. Specifically, when starting with an arbitrary many copies of one nonlocal correlation, it becomes impossible to obtain even a single copy of the other correlation, and this incomparability holds in both directions. Such incomparable quantum correlations can be obtained even in the simplest Bell scenario, which involves two parties, each having two dichotomic measurements setups. Notably, there exist an uncountable number of such incomparable correlations. Our result challenges the notion of a `unique gold coin', often referred to as the `maximally resourceful state', within the framework of the resource theory of quantum nonlocality. To this end, we provide examples of isotropic quantum correlations that cannot be distilled up-to Tsirelson point,and thus partially answer a long standing open question in nonlocality distillation.
\end{abstract}

\maketitle
{\it Introduction.--} John S. Bell's groundbreaking work in 1964 represented one of the most significant departures from classical worldviews within the realm of quantum physics \cite{Bell1964}. His work challenged the deeply ingrained concept of `local causality' \cite{Bell1966, Mermin1993, Brunner2014}. Bell devised an elegant method to establish the nonlocal behavior of input-output correlations obtained in experiments involving multipartite quantum systems. Subsequently, several milestone experiments with entangled photons provided empirical evidence for quantum nonlocality \cite{Freedman1972, Aspect1981, Aspect1982, Aspect1982(1), Zukowski1993}, thereby settling a long-standing debate on the foundations of quantum physics \cite{Einstein1935, Bohr1935, Schrodinger1935}. With the advent of quantum information science, quantum nonlocality has emerged as a valuable resource for various device-independent protocols \cite{Scaran2012, Ekert1991, Barrett2005a, Acin2006, Pironio2010, Colbeck2012, Chaturvedi2015, Mukherjee2015, Brunner2013, Pappa2015, Roy2016, Banik2019}.

Quantifying the extent of nonlocality in correlations obtained from entangled quantum systems, thus, holds significant practical importance. The framework of quantum resource theories (QRTs) provides an elegant approach to investigate this question \cite{Chitambar2019}. A QRT begins by identifying a set of constrained operations called `free operations' and a subset of states referred to as `free states'. States falling outside this category are called `resourceful states' or simply `resources'. A quintessential example of a QRT is the theory of quantum entanglement, where multipartite systems prepared in non-separable states are considered resources under local operation and classical communication (LOCC) \cite{Horodecki2009}. While exploring nonlocality, the focus shifts from multipartite quantum states to multipartite input-output correlations among distant parties. Of particular interest is the broad spectrum of correlations known as no-signalling (NS) correlations, where communication between the parties is strictly prohibited. Notably, within the realm of classical physics, correlations adhere to a more restrictive framework known as Bell-local correlations, which are encompassed within the NS set. Correlations that transcend this local boundary are termed nonlocal correlations. Remarkably, entangled quantum states are capable of producing such nonlocal correlations, which serve as crucial resources for various protocols \cite{Brunner2014}. Within the framework of resource theory, nonlocality is regarded as a resource, subject to the constraints of free operations comprising local operations and shared randomness (LOSR). More generally the set of free operations consists of wirings and classical communication prior to the inputs (WCCPI) \cite{Gallego2012}. 

Once the free operations, free states, and resourceful states are identified in a resource theory, the next crucial question  is to compare the resources in different states. One pertinent approach is to determine the optimal rates at which these states can be successfully inter-converted under free operations. In this study, we investigate the concept of resource inter-convertibility among quantum nonlocal correlations. Firstly, we observe that even in the simplest Bell scenario, involving two spatially separated parties, each conducting two dichotomic measurements, there are uncountably many quantum nonlocal correlations that cannot be freely converted into each other at the single-copy level, highlighting the incomparability of these resources. We then investigate this question by considering asymptotically many copies of these resources. In doing so, we establish an even more striking result. We prove that there exist quantum nonlocal correlations that are inequivalent in the strongest sense, as they are not inter-convertible even under asymptotic manipulation. More particularly, there are quantum correlations $P_q$ and $P^\prime_q$ such that, starting with an arbitrary number of copies of $P_q$, it is not possible to obtain even a single copy of $P^\prime_q$ under the free operation of LOSR, and vice versa. This finding distinguishes the theory of quantum nonlocality from the theory of quantum entanglement. In the case of bipartite entanglement, asymptotic state inter-conversion gives rise to the concepts of entanglement distillation and entanglement cost \cite{Bennett1996(2)} (see also \cite{Hayashi2003}). Consequently, the notion of a maximally entangled state, a `unique gold coin,' emerges. Our result, however, establishes that quantum nonlocal correlations lack the concept of such a unique gold coin, thereby resulting significant implication in the study of nonlocality distillation \cite{Forster2009, Brunner2009, Ebbe2014, Beigi2015, Brito2019, Eftaxias2022, Naik2023}.

{\it Preliminaries.--} The $\mathrm{nmk}$-Bell scenario consists of $n$ distant parties, each performing $m$ different $k$-outcome measurements on their respective subsystems. By repeating the experiments many times they produce a joint input-output correlation $P:=\{p(\Vec{a}|\Vec{x})\equiv p(a_1,\cdots,a_n|x_1,\cdots,x_n)~|~x_i\in\mathcal{X}_i,~a_i\in\mathcal{A}_i\}$, where $|\mathcal{X}_i|=m,~|\mathcal{A}_i|=k,~\forall~i\in\{1,\cdots n\}$. The joint probabilities satisfy the no-signaling (NS) conditions that prohibit instantaneous information transfer among the distant parties. Set of all NS correlations forms a convex polytope $\mathcal{N}$ embedded in some $\mathbb{R}^N$ (the value of $N$ depends on $n,~m,$ and $k$). A correlation is called `Bell local' if it can be factorized as $p(\Vec{a}|\Vec{x})=\int_{\Lambda}d\lambda p(\lambda)\Pi_{i=1}^n p(a_i|x_i,\lambda)$, where $\lambda\in\Lambda$ is a classical variable shared among the parties, and $p(\lambda)$ is a probability density function over $\Lambda$ \cite{Brunner2014}. The set of local correlations forms a proper sub-polytope $\mathcal{L}$. A correlation is called quantum if it allows a quantum realization, {\it i.e.}, $p(\Vec{a}|\Vec{x})=\mbox{Tr}[(\otimes_{i=1}^n\pi^{a_i}_{x_i})\ket{\psi}\bra{\psi}]$, where $\ket{\psi}\in\otimes_{i=1}^n\mathcal{H}_i$ and $\pi^{a_i}_{x_i}\in\mathcal{P}(\mathcal{H}_i)$ with $\sum_{a_i}\pi^{a_i}_{x_i}=\mathbf{I}_{\mathcal{H}_i}$. Dimensions of the Hilbert spaces are finite, {\it i.e.}, $\mbox{dim}(\mathcal{H}_i)<\infty$, and $\mathcal{P}(\star)$ denotes the sets of positive operators acting on the respective Hilbert spaces (see \cite{Cabello2023} for other possible mathematical models for physical correlations). Set of all quantum correlations $\mathcal{Q}$ forms a convex set lying strictly in between the local and NS polytopes, {\it i.e.}, $\mathcal{L}\subsetneq\mathcal{Q}\subsetneq\mathcal{N}$. For the $222$-Bell scenario the polytope $\mathcal{N}$, embedded in $\mathbb{R}^8$, has $16$ local deterministic vertices and $8$ non-local vertices \cite{Barrett2005(0)}: 
\begin{subequations}
\begin{align}
&P_{\text{L}}^{\alpha\beta\gamma\eta}\equiv\left\{p(ab|xy):=\delta_{(a,\alpha x\oplus\beta)}~\delta_{(b,\gamma y\oplus \eta)}\right\};\label{L}\\ 
&P_{\text{NL}}^{\alpha\beta\gamma}\equiv\left\{p(ab|xy):= 1/2~\delta_{(a \oplus b,~xy\oplus\alpha
x\oplus\beta y\oplus\gamma)}\right\};\label{NL}
\end{align}
\end{subequations}
with $\alpha,\beta,\gamma,\eta\in\left\{0,1\right\}$, whereas the polytope $\mathcal{L}$ is the convex hull of local deterministic vertices. The quantum set $\mathcal{Q}$ forms a convex set with uncountably many nonlocal extreme points, each having quantum realization with two-qubit pure entangled state and local projective measurements \cite{Masanes2006}.

In a resource theory, two resources $R_1$ and $R_2$ will be called equivalent, symbolized as $R_1\simm R_2$, if $R_2$ can be obtained from $R_1$ under the free operations and the vice-verse. Collection of equivalent resources form a equivalent class. On the other hand, $R_1\suc R_2$ puts an ordering ``$R_1$ is more resourceful than $R_2$" in the sense that $R_2$ can be obtained from $R_1$ freely but not the other-way around. Finally, $R_1\nsimm R_2$ denotes that neither $R_2$ can be freely obtained from $R_1$ nor $R_1$ from $R_2$. In such a case resources $R_1$ and $R_2$ are incomparable, and hence they are treated as inequivalent resources. For instance, in resource theory of nonlocality, correlations in $\mathcal{L}$ are free states, while those lying in $\mathcal{N}\setminus\mathcal{L}$ are the resources \cite{Vicente2014}. For the $222$ scenario, the extremal nonlocal correlations of Eq.(\ref{NL}) form an equivalence class as they are inter-convertible under local reversible operations \cite{Barrett2005(0)}.  

{\it Results.--} In this work, we consider a physically motivated variant of nonlocality theory which we call the resource theory of quantum nonlocality (RTQN). All the correlations allowed in this theory are quantum-realizable, {\it i.e.}, the resources belong to the set $\mathcal{Q}\setminus\mathcal{L}$. Interestingly, there are extreme points of $\mathcal{Q}$ that are inequivalent under one-copy manipulation -- in-fact there are uncountably many of them (see Proposition 28 in \cite{Wolfe2020}). Equivalent classes of these extreme correlations are discussed in the Supplemental \cite{Supple}. In a generic resource theory, it is quite possible that a resource $R_2$ cannot be obtained from one copy of another resource $R_1$, but can be obtained from its $n$ copies. The symbol $R_1^{\otimes n}\hspace{-.2cm}\nrightarrow\hspace{-.1cm} R_2$ denotes that a single copy of $R_2$ cannot be obtained from $n$ copy of $R_1$ under the allowed free operations. This leads to a notion of the strongest form of inequivalence, namely the asymptotically inequivalence between two resources: \vspace{-.2cm}
\begin{align}\label{asym}
R_1\ansimm R_2,~\mbox{whenever}~R_1^{\otimes n}\hspace{-.2cm}\nrightarrow\hspace{-.1cm} R_2~\&~R_2^{\otimes n}\hspace{-.2cm}\nrightarrow\hspace{-.1cm} R_1,~\forall~n\in\mathbb{N}.
\end{align}
For instance, in bipartite entanglement theory, single-copy inter-convertibility of pure entangled state is completely determined through majorization criteria \cite{Nielsen1999}. While there are pure entangled states that are incomparable according to this criteria, in asymptotic setup all of them become comparable \cite{Bennett1996(2)}. Consequently, the notion of maximally entangled state arises, which for $(\mathbb{C}^d)^{\otimes2}$ system reads as $\ket{\phi^+_d}:=(\sum_{i=0}^{d-1}\ket{ii})/\sqrt{d}$, where $\{\ket{i}\}_{i=0}^{d-1}$ is the computational basis of $\mathbb{C}^d$. In nonlocality scenario, a large class of $\mathrm{nm2}$ NS correlations can be simulated with multiple copies of the $222$ nonlocal vertex, which otherwise are not possible with a single copy \cite{Barrett2005,Nick2005,Sidhardh2023}.  

Therefore, naturally the question arises whether an ordering relation can be reestablished among the extremal quantum correlations under asymptotic manipulation that otherwise are incomparable at the single-copy level. In this work we will, however, show that there are quantum nonlocal correlations that are incomparable even in asymptotic setup. To this aim, we first consider two specific nonlocal extreme points -- the Tsirelson correlation $P_T$ that saturates the maximum quantum value $2\sqrt{2}$ of the Clauser-Horne-Shimony-Holt (CHSH) expression $\mathcal{CHSH}:=\langle X_0Y_0\rangle+\langle X_0Y_1\rangle+\langle X_1Y_0\rangle-\langle X_1Y_1\rangle$ \cite{Clauser1969,Cirelson1980}, and the Hardy correlation $P_H$ that yields the maximum quantum success $(5\sqrt{5}-11)/2\approx0.09$ for the Hardy's argument \cite{Hardy1992}. Quantum realizations for the correlations $P_T$ and $P_H$ are given by
\begin{subequations}
\begin{align}\label{PT}
P_T&\overset{Q}{=}\left\{\!\begin{aligned}
\ket{\phi^+_2}=(\ket{00}+\ket{11})/\sqrt{2};~X_0=\sigma_z,\\X_1=\sigma_x;Y_j=\left(\sigma_z+(-1)^j\sigma_x\right)/\sqrt{2}
\end{aligned}\right\};\\\label{PH}
P_H&\overset{Q}{=}\left\{\!\begin{aligned}
\ket{\psi_H}=a(\ket{01}+\ket{10})+\sqrt{1-2a^2}\ket{11},\\K_0=\sigma_z,~~K_1=\ket{\alpha}\bra{\alpha}-\ket{\alpha^\perp}\bra{\alpha^\perp}~~~
\end{aligned}\right\};
\end{align}
\end{subequations}
where $\ket{\alpha}:=(\sqrt{1-2a^2}\ket{0}-a\ket{1})/\sqrt{1-a^2},~a:=\left((5-\sqrt{3})/2\right)^{1/2}$, and $K\in\{X,Y\}$. To prove the asymptotic inequivalence of $P_T$ and $P_H$ we start by recalling a simple mathematical Lemma from \cite{Jozsa1994} (for the sake of completeness we discuss the proof in Supplementary material \cite{Supple}). 
\begin{Lemma}\label{lemma1}
$(X\otimes Y)\ket{\tilde{\phi}^+_d}=\left(\mathbf{I}_d\otimes YX^{\mathrm{T}}\right) \ket{\tilde{\phi}^+_d}$, where $X,Y\in\mathcal{B}(\mathbb{C}^d)$ and $\mathbf{I}_d$ is the identity operator on $\mathbb{C}^d$.
\end{Lemma}
Here, $\ket{\tilde{\psi}}$ denotes the unnormalized vector corresponds to the state $\ket{\psi}$, $\mathcal{B}(\star)$ denotes the set of bounded operators acting on the corresponding Hilbert space, and `$\mathrm{T}$' denotes transposition in computational basis. We now proceed to prove our first {\it no-go} result on multi-copy manipulation of quantum nonlocal correlations. 
\begin{proposition}\label{prop1}
Even a single copy of the correlation $P_H\in\mathcal{Q}$ cannot be obtained from arbitrary many copies of the correlation $P_T\in\mathcal{Q}$ under LOSR, {\it i.e.}, $P_T^{\otimes n}\hspace{-.2cm}\nrightarrow\hspace{-.1cm} P_H,~\forall~n\in\mathbb{N}$.
\end{proposition}
\begin{proof}
Note that $n$-copies of the correlation $P_T$ can be obtained from the state $\ket{\phi^+_2}^{\otimes n}\equiv\ket{\phi^+_{2^n}}$. On the other hand, any $222$ correlation obtained through LOSR protocols applied on $P_T^{\otimes n}$ can also be obtained by performing two dichotomic measurements on the each local part of the state $\ket{\phi^+_{2^n}}$ \cite{Dukaric2008,Lang2014}. Therefore, to prove the present proposition, it is sufficient to show that the state $\ket{\phi^+_2}^{\otimes n}$ does not exhibit Hardy's nonlocality. Furthermore, we can restrict ourselves to projective measurements, since a dichotomic POVM can always be represent as probabilistic mixture of projective measurements \cite{Masanes2006}. Recall that the Hardy nonlocality argument reads as \cite{Hardy1992}:
\begin{align*}
p(00|X_0Y_0)&=q>0,\\
p(00|X_0Y_1)=p(00|X_1Y_0)&=p(11|X_1Y_1)=0.
\end{align*}
Applying Lemma \ref{lemma1} on $\phi^+_{2^n}\equiv\ket{\phi^+_{2^n}}\bra{\phi^+_{2^n}}$, we have 
\begin{align*}
p(ab|X_iY_j)=&\frac{1}{2^n}\text{Tr}\left[\left\{\mathbf{I}_{2^n}\otimes Y_j^b\left(X_i^a\right)^{\mathrm{T}}\right\}\phi^+_{2^n}\right]\approx\text{Tr}\left(Y_j^b\bar{X}_i^a\right).
\end{align*}
where, $X_i^a(Y_j^b)$ be the projector corresponding to the outcome $a(b)$ of measurement $X_i(Y_j)$, $\bar{X}_i^a:=(X_i^a)^{\mathrm{T}}$, $i,j\in\{0,1\}$, and `$\approx$' denotes the unnormalized probability value. Plugging these expressions in Hardy's argument we get
\begin{subequations}
\begin{align}
\label{h1}
\text{Tr}\left(Y_0^0\bar{X}_0^0\right)>0&\implies\text{Supp}\left(Y_0^0\right)\cap\text{Supp}\left(\bar{X}_0^0\right)\neq\emptyset,\\ \label{h2}
\text{Tr}\left(Y_1^0\bar{X}_0^0\right)=0&\implies\text{Supp}\left(\bar{X}_0^0\right)\subseteq\text{Supp}\left(Y_1^1\right),\\ \label{h3}
\text{Tr}\left(Y_0^0\bar{X}_1^0\right)=0&\implies\text{Supp}\left(\bar{X}_1^0\right)\subseteq\text{Supp}\left(Y_0^1\right),\\ \label{h4}
\text{Tr}\left(Y_1^1\bar{X}_1^1\right)=0&\implies\text{Supp}\left(Y_1^1\right)\subseteq\text{Supp}\left(\bar{X}_1^0\right),
\end{align}
\end{subequations}
where, $\text{Supp}(Z)\subseteq\mathbb{C}^{2^n}$ denotes the support of the projector $Z$. Eqs.(\ref{h2}), (\ref{h3}), \& (\ref{h4}) imply $\text{Supp}\left(\bar{X}_0^0\right)\subseteq\text{Supp}\left(Y_0^1\right)$. On the other hand, $Y_0$ being a projective measurement implies  $\text{Supp}\left(Y_0^1\right)\cap\text{Supp}\left(Y_0^0\right)=\emptyset$, which thus forbids the condition (\ref{h1}) to be held true. This  completes the proof.
\end{proof}
Important to note that Proposition \ref{prop1} holds true even if the correlation $P_H\in\mathcal{Q}$ is replaced by other Hardy's correlations $P_h\in\mathcal{Q}$ arising from other two-qubit non-maximally entangled states \cite{Goldstein1994,Kar1997,Rai2022}, where the success probability is less than the quantum optimal value, {\it i.e.}, $0<p_h(00|X_0Y_0)<p_H(00|X_0Y_0)=(5\sqrt{5}-11)/2$. We now proceed to address the reverse-way inter-conversion of the resources appeared in Proposition \ref{prop1}. Furthermore, it is also important to note that this Proposition as well as the other results obtained in this work also holds true if we consider the set of more general free operations WCCPI, instead of LOSR (argument provided in \cite{Supple}).  
\begin{proposition}\label{prop2}
Even a single copy of the correlation $P_T\in\mathcal{Q}$ cannot be obtained from arbitrary many copies of the correlation $P_H\in\mathcal{Q}$ under LOSR, {\it i.e.}, $P_H^{\otimes n}\hspace{-.2cm}\nrightarrow\hspace{-.1cm} P_T,~\forall~n\in\mathbb{N}$.
\end{proposition}
\begin{proof}
The correlation $P_H$ has the quantum realization of Eq.(\ref{PH}). Therefore $n$-copy of the correlation $P_H^{\otimes n}$ can be obtained from the quantum state $\ket{\psi_H}_{AB}^{\otimes n}\in(\mathbb{C}_A^2\otimes\mathbb{C}_B^2)^{\otimes n}$. Since any $222$ correlation obtained through LOSR protocol on $P_H^{\otimes n}$ can be obtained by performing two dichotomic measurements on each part of the state  $\ket{\psi_H}_{AB}^{\otimes n}$, and since the correlation $P_T$ self-test the quantum state $\ket{\phi^+_2}$ \cite{Supic2020}, therefore contrary to the claim of the Proposition if we assume $P_H^{\otimes n}\hspace{-.2cm}\rightarrow\hspace{-.1cm} P_T$, then we must have
\begin{align}\label{nht0}
\Phi_A\otimes\Phi_B(\ket{\psi_H}_{AB}^{\otimes n})=\ket{\phi_2^+}_{A_1B_1}\ket{\zeta}_{A_2B_2},
\end{align}
where $\Phi_D:(\mathbb{C}^2)^{\otimes n}_D\mapsto\mathbb{C}^2_{D_1}\otimes\mathcal{H}_{D_2}$ be the isometric maps for $D\in\{A,B\}$, which can be thought as unitary by incorporating ancillary systems, {\it i.e.},
\begin{align}\label{nHT}
U_A\otimes U_B(\ket{\psi_H}_{AB}^{\otimes n}\ket{\eta}_{A^\prime}\ket{\eta}_{B^\prime})=\ket{\phi_2^+}_{A_1B_1}\ket{\zeta}_{A_2B_2},
\end{align}
where, the local ancillary states $\ket{\eta}_{A^\prime}~\&~\ket{\eta}_{B^\prime}$ are taken to make the input and output Hilbert spaces to be of same dimension. An immediate consequence is that the eigenvalues (EV) of the reduced part of the states on the left and right sides of Eq.(\ref{nHT}) must be same, {\it i.e.},
\begin{align}\label{evs}
\text{EV}\left\{(\rho^{\psi_H})^{\otimes n}_A\otimes\ket{\eta}_{A^\prime}\bra{\eta}\right\}\equiv \text{EV}\left\{(\mathbf{I}_2)_{A_1}/2\otimes\rho^\zeta_{A_2}\right\},
\end{align}
where $\rho_A^\chi$ denotes $A$ subsystem's marginal state of the composite state $\ket{\chi}_{AB}$. Let Schmidt coefficients of the state $\ket{\psi_H}$ be $\{\sqrt{s},\sqrt{1-s}\}$. The EVs on the left hand part of Eq.(\ref{evs}) are
\begin{align*}\label{evl}
EV\{L\}\equiv\left\{s^{n},s^{(n-1)}(1-s),\cdots,(1-s)^{n},0,\cdots,0\right\}.    
\end{align*}
On the other hand, for the right hand part of Eq.(\ref{evs}) the nonzero eigenvalues are evenly degenerate. Therefore, a necessary condition to hold Eq.(\ref{evs}) is that $s^n=s^{(n-j)}(1-s)^j$ for some $j\in\{1,\cdots,n\}$. However, this implies $s=1/2$, a contradiction, and hence completes the proof.  
\end{proof}
Importantly, Proposition \ref{prop2} holds true for any pairs of quantum correlations $P_{\phi^+_2}, P_\psi$,  where the correlation $P_{\phi^+_2}$ self-tests the state $\ket{\phi^+_2}$ and the correlations $P_\psi$ allow quantum realization with two-qubit non-maximally entangled states $\ket{\psi}$, not necessarily self-tests the state $\ket{\psi}$ and neither being an extremal quantum correlation; and thus we have $P_\psi^{\otimes n}\hspace{-.2cm}\nrightarrow\hspace{-.1cm} P_{\phi^+_2},~\forall~n\in\mathbb{N}$. While examples of $P_\psi$ can be constructed immediately, $P_{\phi^+_2}$ are the Tsirel'son-Landau-Masanes (TLM) boundary points of $222$ correlations \cite{Tsirel'son1987,Landau1988,Masanes2003}. Proceeding further, Proposition \ref{prop1} and Proposition \ref{prop2} lead us to the following theorem.
\begin{theorem}\label{theo1}
The quantum correlations $P_T$ and $P_H$ are incomparable in the strongest sense, {\it i.e.}, $P_T\ansimm P_H$.    
\end{theorem}
A comparative discussion with entanglement theory is worthwhile at this point. For the bipartite case, all the pure entangled states can be compared under LOCC. In fact, the von Neumann entropy of the reduced part of such states uniquely quantifies their entanglement. Theorem \ref{theo1}, in this sense, distinguishes RTQN from the theory of quantum entanglement. Importantly, the existence of bound entangled states with negative partial transpose (NPT) will lead to bipartite mixed entangled states that are incomparable in the strongest sense \cite{Vollbrecht2002}. However, the existence of such strongly incomparable pairs of mixed entangled states does not necessitate the existence of bound NPT states. Two entangled states with positive partial transposition might also serve as an example. Albeit we do not know example of any such pair of states. One may wonder whether the inequivalence established in Theorem \ref{theo1} is specific to the pair of correlations $P_T$ and $P_H$, and then having a "gold coin" (other than $P_T$ and $P_H$) can not be ruled out immediately. Nevertheless, our next result shows that there are uncountably many such inequivalent pairs of extreme nonlocal correlations and consequently lead us to conclude about the non-existence of any gold-coin resource.
\begin{theorem}\label{theo2}
All the pairs of quantum correlations $P_{\phi^+_2},P^{st}_{\psi}$ are incomparable in the strongest sense, {\it i.e.}, $P_{\phi^+_2}\ansimm P^{st}_{\psi}$.
\end{theorem}
Here $P_{\phi^+}$'s self-test the state $\ket{\phi^+_2}$ and $P^{st}_{\psi}$'s self-test the two-qubit non-maximally entangled states $\ket{\psi}$. The proof of this theorem is similar to Proposition \ref{prop2}. For the sake of completeness we discuss the proof in Supplementary material \cite{Supple}. 

{\it Distilling Nonlocality.--} In nonlocality distillation, the goal is to obtain highly nonlocal correlations by starting with multiple copies of weakly nonlocal systems \cite{Forster2009, Brunner2009, Ebbe2014, Beigi2015, Brito2019, Eftaxias2022, Naik2023}. As a consequence of the above theorems, we will now derive a nontrivial restriction on the asymptotic distillation of nonlocal quantum correlations.
\begin{corollary}\label{coro1}
Consider the correlations $P^{(\lambda)}_x,P^{(\lambda)}_y\in\mathcal{Q}$, such that $P^{(\lambda)}_x:=\lambda P_{x}+(1-\lambda)L$ and $P^{(\lambda)}_y:=\lambda P_{y}+(1-\lambda)L$ with $x\in X\equiv\{H,\psi\}$, $y\in Y\equiv\{T\}$, and $\lambda\in(0,1]$. Starting with arbitrary many copies of the correlation, neither $P^{(\lambda)}_x$ can be distilled to $P_y$ nor $P^{(\lambda)}_y$ can be distilled to $P_x$.
\end{corollary}
\begin{proof}
$N$-copies of the correlation $[P^{(\lambda)}_x]^{\otimes N}$ reads as
$[P^{(\lambda)}_x]^{\otimes N}=\sum_{k=0}^{N}\lambda^k(1-\lambda)^{(N-k)} \times\Pi_k\{P_x^{\otimes k}\otimes L^{\otimes (N-k)}\}$, where $\Pi_k\{P_x^{\otimes k}\otimes L^{\otimes (N-k)}\}$ denotes all possible permutations of $k$ copies of $P_x$ and $(N-k)$ copies of $L$. Note that, sharing any local box is allowed as free operation within the resource theory of nonlocal. On the other hand, Theorem \ref{theo1} and \ref{theo2} imply $P_x\ansimm P_y,~\forall~x\in X\text{ and }y\in Y$. Finally, noting that the similar decomposition also holds for $[P^{(\lambda)}_y]^{\otimes N}$ we thus establish the claim.
\end{proof}
Consider the class of $222$ isotropic correlations defined as,
\begin{align*}
PR_{\eta}(ab|xy):=\begin{cases}(1+\eta)/4,~~~\mbox{if}~a\oplus b=xy\\
(1-\eta)/4,~~~~~~~~~\mbox{otherwise}.\end{cases}
\end{align*}
For $0\le\eta\le1$ the correlations belong to the set $\mathcal{N}$, for $0\le\eta\le1/\sqrt{2}$ they belong to $\mathcal{Q}$, and for $0\le\eta\le1/2$ they belong to $\mathcal{L}$. Furthermore, $PR_{1/\sqrt{2}}$ corresponds to the Tsirelson's  $P_T$. A well-known conjecture, regarding distillability of isotropic correlations is that from arbitrary many copies of $PR_{\eta_1}$ it is not possible to distill $PR_{\eta_2}$, where $1/2<\eta_1<\eta_2<1$ \cite{Lang2014}. While some partial results are known with finite copy manipulation \cite{Short2009,Forster2011}, recently the authors in \cite{Beigi2015} have proved the conjecture for correlations with $1/\sqrt{2}<\eta_1<\eta_2<1$. Our next theorem establishes a nontrivial result to this direction with isotropic quantum nonlocal correlations.
\begin{theorem}\label{t3}
There exist isotropic quantum correlations $PR_\eta$, with $\eta\in(1/2,1/\sqrt{2})$, that cannot be distilled up to $P_T$, even asymptotically.
\end{theorem}
Proof of the Theorem just follows from Corollary \ref{coro1} and the geometry of correlation space (see Fig. \ref{fig1}).

At this point, one may pose a different question. While nonlocality distillation is typically motivated by the desired resource one wishes to achieve, there might be protocols that simultaneously distill fractions of different inequivalent resources. In other words, the absence of a unique `gold coin resource' does not immediately rule out the existence of such a `gold protocol' (see the Supplementary material \cite{Supple} for pictorial explanation). At present we do not know any analytic method to tackle this question, and hence leave this question for future research.
\begin{figure}[t!]
\includegraphics[width=0.48\textwidth,left]{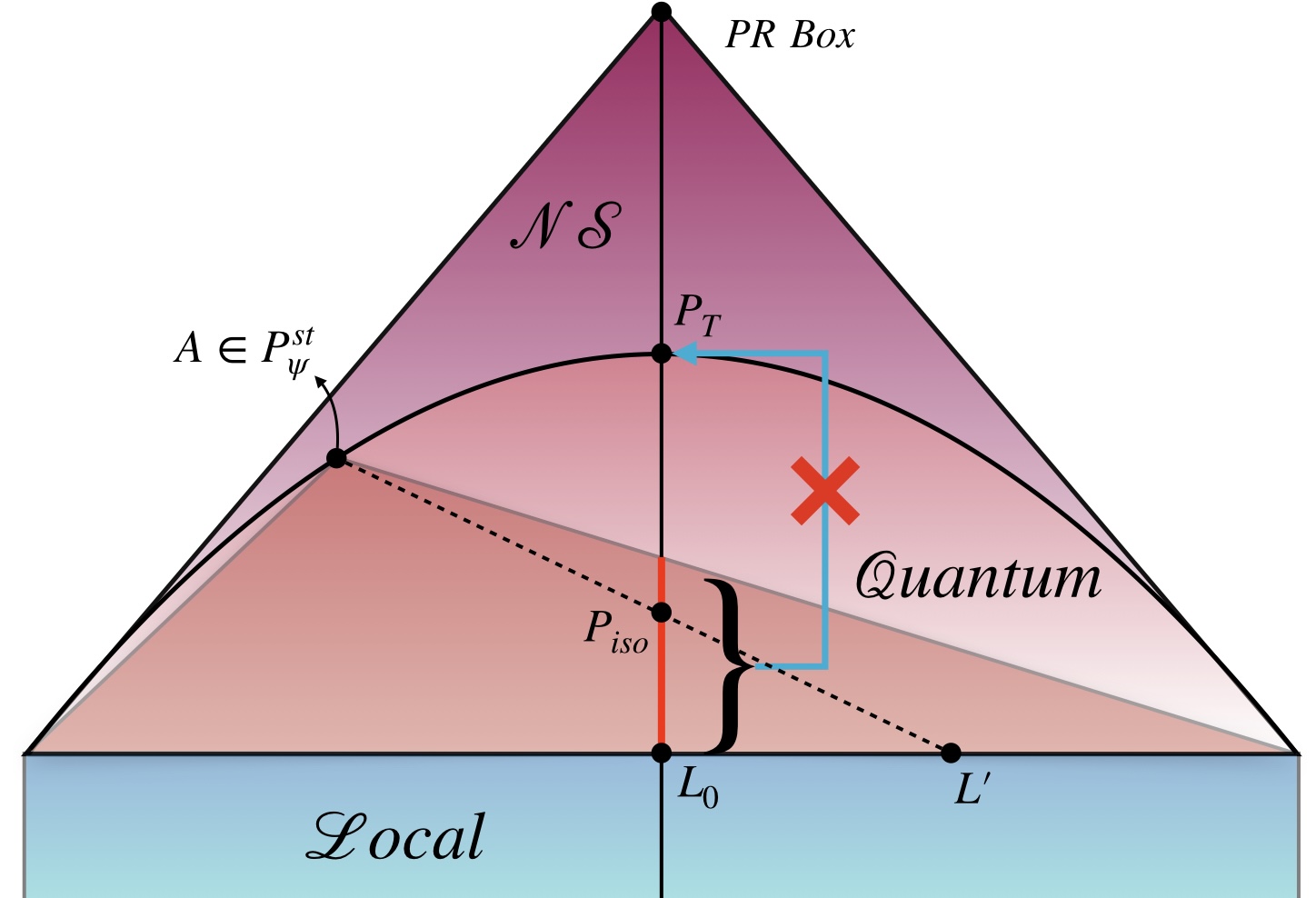}
\caption{According to Corollary \ref{coro1}, any quantum point obtained through convex mixing of the point $P_{\psi}^{st}$ and a local point $L^\prime$ ({\it i.e.}, the line AL$^\prime$) cannot be distilled to the point $P_T$, even asymptotically. The correlations on the line joining the points PR-Box and $L_0$ are isotropic correlations, with the line segment $(L_0,P_T]$ representing quantum nonlocal isotropic correlations ({\it i.e.}, $1/2<\eta<1/\sqrt{2}$). Clearly, the quantum isotropic correlation $P_{iso}$ cannot be distilled to $P_T$, partially solving the conjecture proposed in \cite{Lang2014} within the quantum region, as stated in our Theorem \ref{t3}. By varying the point $L^\prime$, we can, in fact, obtain a set of quantum isotropic correlations with nonzero measure that cannot be distilled to $P_T$.}
\label{fig1}
\end{figure}

{\it Discussions.--} Establishing asymptotic inequivalence among different types of quantum nonlocal correlations carries significant practical implications. These correlations are pivotal for various information-theoretic tasks. Our Theorem 1 elucidates that if a specific quantum correlation is indispensable for the flawless execution of a task, then the same task may not be executed flawlessly even with numerous copies of inequivalent quantum correlations. Instances of such scenarios have been documented in zero-error and reverse-zero-error communication scenarios \cite{Cubitt2010,Cubitt2011,Alimuddin2023}, as well as in Bayesian game scenarios \cite{Banik2019}. Consequently, when deriving nonlocal correlations from entangled quantum states for these tasks, it is imperative to perform the appropriate local measurements on the given state.  

It is crucial to highlight that in our investigation, we have presumed that both the quantum state and measurement devices are predetermined, thereby resulting in nonlocal correlations that can be subsequently altered through the free operation of Local Operations and Shared Resources (LOSR). However, an alternative scenario can be envisaged, wherein the local components of multiple copies of these states are collectively manipulated by conducting measurements in an entangled basis. This scenario gives rise to a distinct resource theory, namely the resource theory of entanglement under LOSR. As we note that, in this broader framework, it is possible, albeit probabilistically, to obtain a correlation $P_T$ starting with many copies of $P_H$.  However, we are unaware of any protocol that yields the correlation $P_H$ starting with many copies of $P_T$. Asymptotic analysis of probabilistic transformation among different nonlocal correlations, in this broader framework, promises to shed light on the intricate structures of quantum nonlocal correlations and quantum entanglement. Finally, while our results establish intricacies in multi-copy manipulation of quantum nonlocal correlations, the present study mainly deals with the $222$ correlations. A similar analysis for multipartite correlations with higher number of inputs and outputs is worth exploring.  

\begin{acknowledgements}
SRC acknowledges support from University Grants Commission, India (reference no. 211610113404). MB acknowledges funding from the National Mission in Interdisciplinary Cyber-Physical systems from the Department of Science and Technology through the I-HUB Quantum Technology Foundation (Grant no: I-HUB/PDF/2021-22/008), support through the research grant of INSPIRE Faculty fellowship from the Department of Science and Technology, Government of India, and the start-up research grant from SERB, Department of Science and Technology (Grant no: SRG/2021/000267).
\end{acknowledgements}

\section{Supplemental}
\section{Framework of Nonlocal Resource Theory}
Every resource theory begins by establishing the concepts of free resources and free operations \cite{Chitambar2019}. In this section, we provide a brief overview of the resource theory of nonlocality. \\\\
\textbf{Free resources:} In resource theory of nonlocality we are interested in the joint input-output statistics among spatially separated parties. Space-like separation among the parties demands joint probabilities to satisfy the no-signaling condition that prohibits instantaneous transfer of information among the different parties. In this resource theory, the free resources are the Bell local correlations that can be factorized as 
\begin{align}
p(\Vec{a}|\Vec{x})=\int_{\Lambda}d\lambda p(\lambda)p(a_1|x_1,\lambda)\cdots p(a_n|x_n,\lambda),   
\end{align}
where $\lambda\in\Lambda$ is a classical variable shared among the parties, and $p(\lambda)$ is a probability density function over $\Lambda$. Correlations that are not local are termed as nonlocal, and they are the resource. \\
\textbf{Free operations:} 
The operations that keeps free resources free are identified as the free operations. In the context of nonlocal resource theory, such operations are classified as LOSR \cite{Vicente2014}. However, it has been observed that the set of free operations in nonlocal resource theory extends beyond LOSR, encompassing a broader category known as WCCPI \cite{Gallego2012}. In the bipartite Bell scenario, the relation between LOSR and WCCPI has been extensively studied by Gallego et al.  \cite{Gallego2017}. They showed that, action of the operation WCCPI on any no-signaling box $P_{0}$, can be written as 
\begin{align}
\mathcal{W}_{WCCPI}(P_{0}):= pL(P_{0})+(1-p)\mathcal{W}_{LOSR}(P_{0}), 
\end{align}
Where $L(P_{0})\in \mathcal{L}$ for all $P_{0}\in \mathcal{N}$. \\\\ So, the operation of $\mathcal{W}_{WCCPI}$ on a given no-signaling box can be understood as a convex mixture of a local box and another no-signaling box resulting from certain LOSR applied to the original one. It's important to note that in the realm of bipartite nonlocality distillation, the broad spectrum of free operations encompassed by WCCPI doesn't offer any advantage over LOSR operations. 

\section{Single-copy Manipulation of Nonlocal Correlations in $\mathcal{Q}$}
At single copy level, WCCPI protocol simply boils down to relabeling of the inputs and manipulating the outputs based on the given input. More explicitly, starting with a correlation $P(a^\prime b^\prime|x^\prime y^\prime)$ (say parent correlation) the players can decide to choose the input $z$ which is either $z^\prime$ or $\bar{z}^\prime$, where $z=x,y$ for Alice and Bob respectively. While dealing with local correlations the players can choose $z$ to be a constant function. However, such choices will map a nonlocal correlation to a local one \cite{Fine1982}. Upon getting an outcome $a^\prime$ and $b^\prime$ they can determine their final output as $a = f_{A}(x,a')$ and $b = f_{B}(y,b')$, respectively, where $f_{A/B}:\{0,1\}^{\times 2}\mapsto\{0,1\}$ (see Fig.\ref{fig2}). The process will result in a child correlation $\tilde{P}(ab|xy)$. The WCCPI protocol obtained by aforesaid procedure will be called a deterministic protocol. One can consider more general kind of protocols that stochastically map the primed variables to the unprimed ones. In the following we show that starting with single copy of a nonlocal correlation only a finite number of child correlation can be obtained under deterministic WCCPI.
\begin{figure}[h!]
\centering
\includegraphics[scale = 0.12]{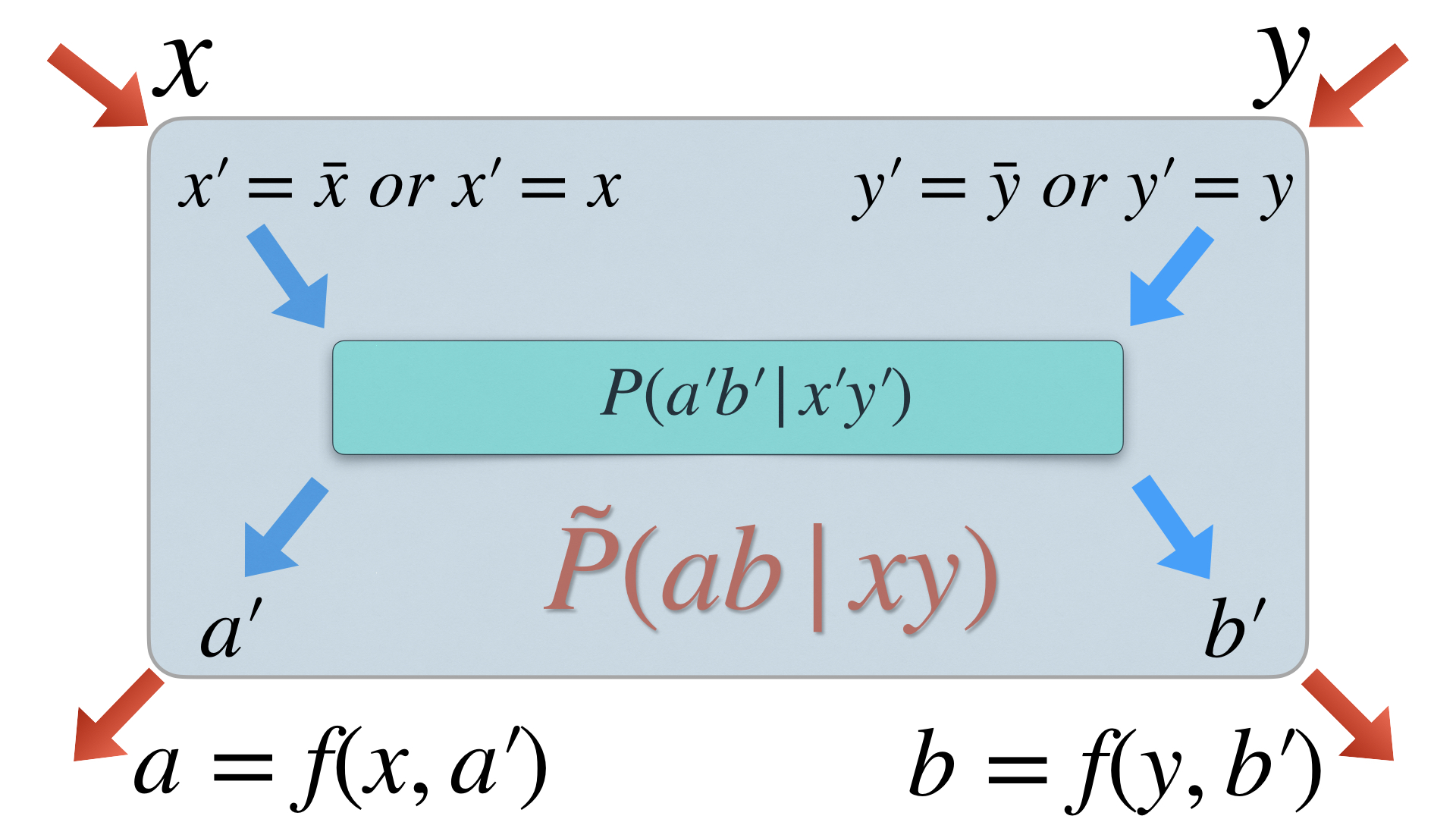}
\caption{Single-copy manipulation:  WCCPI protocol on parent correlation $P(a'b'|xy)$ results in a child correlation $\tilde{P}(ab|xy)$.}
\label{fig2}
\end{figure}
\begin{Lemma}
Deterministic WCCPI protocols on single copy of a $222$ nonlocal correlation can yield at most $8$ different child nonlocal correlations.  
\end{Lemma}
\begin{proof}
As already mentioned to ensure the nonlocality of the resulting child correlation the input has to be $z=z^\prime$ or $z=\bar{z}^\prime$ for $z\in\{x,y\}$. For a given choice of input there exist $16$ different choices for $f_A$ (as well as for $f_B$) to produce the outputs of the child box. However, one can impose further restriction on these choice to ensure the nonlocality. For instance, let for the input $x=0$, the function $f_A(x,a')$ is such that it always yields an unique outcome, say, $a=0$ (see table \ref{tab1}). This implies that the value of $a$ is deterministic for the input $x=0$. Hence a joint probability distribution exists for the outcomes  corresponding to the inputs $x=0$ and $x=1$, and according to the result of \cite{Fine1982} the child becomes local. Therefore such functions are not allowed for our purpose. It does boil down to count the number of such functions that will not result in joint probability distribution of the outputs.    
\begin{table}[t!]
\begin{subtable}[]{0.2\textwidth}
\centering
\begin{tabular}{|c||c|c|}
\hline
 $~~~~z~~~~$ & $~~~~c^\prime~~~~$ & $~~~~c~~~~$ \\
 \hline \hline
 0  & 0  & 0  \\\hline
 0  & 1  & 0  \\\hline
 1  & 0  & 1  \\\hline
 1  & 1  & 0  \\\hline
\end{tabular}
\caption{\centering $c:=z\wedge\bar{c}^\prime$}
\end{subtable}
\hfill
\begin{subtable}[]{0.2\textwidth}
\centering
\begin{tabular}{|c||c|c|}
\hline
$~~~~z~~~~$ & $~~~~c^\prime~~~~$ & $~~~~c~~~~$ \\
\hline \hline
0  & 0  & 0  \\\hline
0  & 1  & 1  \\\hline
1  & 0  & 1  \\\hline
1  & 1  & 0  \\\hline
\end{tabular}
\caption{\centering $c:=z\oplus c^\prime$}
\end{subtable}
\vspace{-.2cm}
\caption{(a) Truth table for the function $c=f(z,c^\prime):=z\wedge \bar{c}^\prime$. This output mapping allows joint probability distribution rendering the child correlation to be local. (b) $c=f(z,c^\prime):=z\oplus c^\prime$: an allowed output mapping.}
\label{tab1}
\end{table}
A careful observation reveals that only the following four output mappings are allowed: 
\begin{align*}
f^1(c,z) &= c,~~~~~~~~~f^2(c,z) = \bar{c},\\
f^3(c,z) &= c\oplus z,~~f^4(c,z)= \overline{c\oplus z},   
\end{align*}
where $(c,z)\in\{(a,x),(b,y)\}$. Now, the relabeling of inputs gives rise to four different scenarios, two for each players.  Together with the $4\times4 = 16$ aforementioned output manipulations, we thus have $64$ different possible single-copy WCCPI protocols that preserves nonlocal character. Among theses $64$ protocols only $8$ map a particular $PR$ box, $P_{\text{NL}}^{\alpha\beta\gamma}$(see Eq. (1b) in the main text) to itself. These protocols are listed in TABLE \ref{tab2}. 
\begin {center}
\begin {table}[h!]
\begin {tabular} {|c|c|c|c|}
\hline
$x'=\bar x,~y'=\bar y$ & $x'=\bar x,~y'=y$ & $x'=x,~y'=\bar y$ & $x'=x,~y'=y$ \\\hline\hline
&&&\\
$a=a'\oplus x$ & $a=a'$ & $a=a'\oplus x$ & $a=a'$\\
$b=\overline{b'\oplus y}$ & $b=b'\oplus y$ & $b=b'$ & $b=b'$ \\\hline
&&&\\
$a = \overline{a'\oplus x}$& $a=\overline{a'}$& $a=\overline{a'\oplus x}$& $a=\overline{a'}$ \\
$b=b'\oplus y$ & $b=\overline{b'\oplus y}$ & $b=\overline{b'}$ & $b=\overline{b}$ \\\hline
\end {tabular}
\caption {$8$ deterministic WCCPI protocols that maps the correlation $P^{\alpha\beta\gamma}_{NL}$ to itself. }
\label{tab2}
\end {table}
\end {center}
Since a stochastic WCCPI protocol cannot preserve extremality of quantum nonlocal correlations, there are only $8$ WCCPI protocols that can map an extremal quantum nonlocal correlation to another extremal. This completes the proof. 
\end{proof}

\section{Proof of Lemma 1}
\begin{proof}
We have 
\begin{subequations}
\begin{align}
\ket{\tilde{\phi}^+_d}&=\sum_{m=0}^{d-1}\ket{m}\otimes\ket{m}\in\mathbb{C}^d\otimes\mathbb{C}^d,\\
A&=\sum_{i,j}a_{ij}\ketbra{i}{j}\in\mathcal{L}(\mathbb{C}^d),\\
B&=\sum_{k,l}b_{kl}\ketbra{k}{l}\in\mathcal{L}(\mathbb{C}^d),
\end{align}
\end{subequations}
where $\mathcal{L}(\star)$ denotes the set of linear operators acting on the corresponding Hilbert space. We therefore have,
\begin{align}
(A\otimes B)\ket{\tilde{\phi}^+_d}&=\sum_{i,j,k,l,m}a_{ij}b_{kl}\ket{i}\langle j|m\rangle\otimes\ket{k}\langle l|m\rangle\nonumber\\
&=\sum_{i,j,k,l,m}a_{ij}b_{kl}\ket{i}\delta_{jm}\otimes\ket{k}\delta_{lm}\nonumber\\
&=\sum_{i,k,m}a_{im}b_{km}\ket{i}\otimes\ket{k}.\label{c1}
\end{align}
We also have,
\begin{align*}
B(A)^{\mathrm{T}}&=\left(\sum_{k,l}b_{kl}\ketbra{k}{l}\right)\left(\sum_{i,j}a_{ij}\ketbra{i}{j}\right)^{\mathrm{T}}\\
&=\sum_{i,j,k,l}a_{ij}b_{kl}\ket{k}\langle l|j\rangle\bra{i}\\
&=\sum_{i,j,k}a_{ij}b_{kj}\ket{k}\bra{i},
\end{align*}
which further implies,
\begin{align}
\left(\mathbf{I}\otimes BA^\mathrm{T}\right)\ket{\tilde{\phi}^+_d}&=\sum_{i,j,k,m}a_{ij}b_{kj}\ket{m}\otimes\ket{k}\langle i|m\rangle\nonumber\\
&=\sum_{i,j,k}a_{ij}b_{kj}\ket{i}\otimes\ket{k}.\label{c2}
\end{align}
The claim of the Lemma 1 follows by comparing Eq.(\ref{c1}) and Eq.(\ref{c2}).
\end{proof}
\begin{figure}[t!]
\includegraphics[width=0.48\textwidth,left]{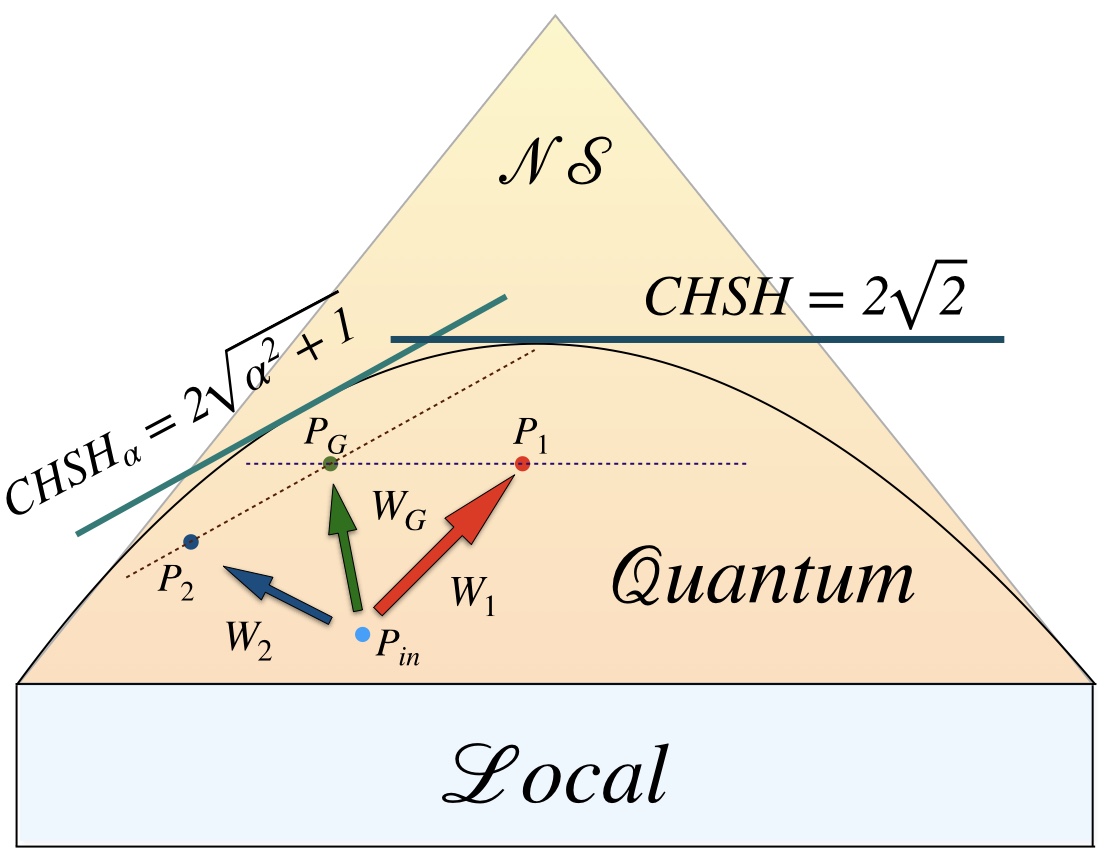}
\caption{{\bf [Possibility of Gold protocol]} Given a correlation $P_{in}\in\mathcal{Q}$, the perpendicular distance from the $\mathcal{CHSH}=2\sqrt{2}$ hyperplane determine the Tsirelson fraction of the correlation, whereas perpendicular distance from the $\mathcal{CHSH}_\alpha:=\alpha\langle X_0\rangle+\mathcal{CHSH}=2\sqrt{\alpha^2+1}$ hyperplane determine its tilted CHSH fraction \cite{Acin2012}. On $P^{\otimes2}_{in}$ the wirings $\mathcal{W}_1$ and $\mathcal{W}_2$ result in a correlations $P_1$ and $P_2$ respectively, having optimal Tsirelson fraction and tilted CHSH fraction. Interestingly, the two-copy gold protocol $\mathcal{W}_G$ optimize both the fractions simultaneously.}
\label{fig1}
\end{figure}

\section{Proof of Theorem 2}
\begin{proof}
To establish the claim we need to prove that
\begin{subequations}
\begin{align}\label{theo2-1}
\left(P^{st}_{\psi}\right)^{\otimes n}\hspace{-.2cm}&\nrightarrow\hspace{-.1cm} P_{\phi^+_2},~\forall~n\in\mathbb{N},~\mbox{and}\\\label{theo2-2}
\left(P_{\phi^+_2}\right)^{\otimes n}\hspace{-.2cm}&\nrightarrow\hspace{-.1cm} P^{st}_{\psi},~\forall~n\in\mathbb{N}.
\end{align}
\end{subequations}
Proof of (\ref{theo2-1}) follows the exact same arguments of Proposition 2. Thus we are remain to prove only (\ref{theo2-2}). Analogous to Eq.(5) of the main text, in this case we have
\begin{align}\label{nmnm}
\Phi_A\otimes\Phi_B(\ket{\phi^+_2}_{AB}^{\otimes n})=\ket{\psi}_{A_1B_1}\ket{\zeta}_{A_2B_2},
\end{align}
Incorporating local ancillas $\ket{\eta}_{A^\prime}\ket{\eta}_{B^\prime}$, the Eq.(\ref{nmnm}) reads as 
\begin{align}\label{nmnm1}
U_A\otimes U_B(\ket{\phi^+_2}_{AB}^{\otimes n}\ket{\eta}_{A^\prime}\ket{\eta}_{B^\prime})=\ket{\psi}_{A_1B_1}\ket{\zeta}_{A_2B_2},
\end{align}
With a similar line of argument as of Proposition 2, we can thus conclude 
\begin{align}\label{nmnm2}
\text{EV}\left\{\left(\mathbf{I}_2/2\right)_A^{\otimes n}\otimes\ket{\eta}_{A^\prime}\bra{\eta}\right\}\equiv \text{EV}\left\{\rho^{\psi}_{A_1}\otimes\rho^\zeta_{A_2}\right\}.
\end{align}
Please note that 
\begin{align*}
\text{EV}\left\{(\mathbf{I}_2)_A/2^{\otimes n}\otimes\ket{\eta}_{A^\prime}\bra{\eta}\right\}\equiv\left\{\frac{1}{2^n},\cdots,\frac{1}{2^n},0,\cdots,0\right\}.  
\end{align*}
This induces a contradiction as $\ket{\psi}$ be a two-qubit non-maximally pure entangled state. This completes the proof. 
\end{proof}
We have already pointed out that apart form $P_T$ the TLM boundary points of $222$ correlations also self-test the state $\ket{\phi^+_2}$ and thus be the bona-fide candidates for $P_{\phi^+_2}$. For the $P^{st}_\psi$ we can consider the correlation that maximize the tilted-CHSH expression 
$\mathcal{CHSH}_\alpha:=\alpha\langle X_0\rangle+\mathcal{CHSH}$ \cite{Acin2012}. Quantum realization of $P^{st}_{\psi}$ is given by,
\begin{align}
\left\{\!\begin{aligned}
\ket{\psi_\theta}=\cos\theta\ket{00}+\sin\theta\ket{11};~X_0=\sigma_z,\\X_1=\sigma_x;~Y_j=\cos\mu~\sigma_z+(-1)^j\sin\mu~\sigma_x
\end{aligned}\right\},
\end{align}
where $\alpha:=2/\sqrt{1+\tan^22\theta},~\tan\mu=\sin2\theta$. 

\end{document}